\documentclass[11pt]{article}
 
\usepackage[margin=1in]{geometry}
\usepackage[utf8]{inputenc}
\usepackage{microtype}
\usepackage{amsmath,amsthm,amssymb, color, tikz, mathtools}

\usetikzlibrary{positioning}
\usetikzlibrary{shapes}
\usetikzlibrary{decorations.pathreplacing}
\usepackage{enumerate} 
\usepackage[pagebackref]{hyperref} 
\hypersetup{colorlinks,linkcolor={blue},citecolor={blue},urlcolor={blue}}
\usepackage{algorithm}
\usepackage{algpseudocode}

\DeclareMathOperator*{\argmax}{arg\,max}

\newcommand{\II}{\mathbb{I}}

\newcommand{\cA}{\mathcal{A}}

\newcommand{\cR}{\mathcal{R}}

\newcommand{\sD}{\mathsf{D}}

\newcommand{\sQ}{\mathsf{Q}}
\newcommand{\sR}{\mathsf{R}}

\newcommand{\set}[1]{\left\{ #1 \right\}}

\renewcommand{\tilde}{\widetilde}

\newcommand{\floor}[1]{\left\lfloor #1 \right\rfloor}

\newcommand{\polylog}{\mathrm{polylog }}
\newcommand{\polyloglog}{\mathrm{polyloglog }}

\usepackage{bbm}

\newcommand{\zone}{\set{0,1}}

\newcommand{\king}{\mathsf{KING}}
\newcommand{\usearch}{\mathsf{USEARCH}}

\newtheorem{theorem}{Theorem}[section]
\newtheorem{corollary}[theorem]{Corollary}

\newtheorem{lemma}[theorem]{Lemma}
\newtheorem{claim}[theorem]{Claim}
\newtheorem{defi}[theorem]{Definition}

\newtheorem{observation}[theorem]{Observation}

\newcommand{\ket}[1]{|#1\rangle}

\newcommand{\cbra}[1]{\left\{#1\right\}}

\renewcommand{\epsilon}{\varepsilon}
\newcommand{\MOD}{\mathsf{MOD}}

\allowdisplaybreaks

\title{Randomized and quantum query complexities of finding a king\\ in a tournament}
\author{Nikhil S.~Mande\thanks{University of Liverpool, UK. {\tt nikhil.mande@liverpool.ac.uk}}
\and 
Manaswi Paraashar\thanks{University of Copenhagen, Denmark. 
{\tt manaswi.isi@gmail.com}
}
\and
Nitin Saurabh\thanks{Indian Institute of Technology, Hyderabad, India. 
{\tt nitin@cse.iith.ac.in}
}
}
\date{}

\begin{document}

\maketitle

\begin{abstract}
    A tournament is a complete directed graph. It is well known that every tournament contains at least one vertex $v$ such that every other vertex is reachable from $v$ by a path of length at most $2$. All such vertices $v$ are called \emph{kings} of the underlying tournament. Despite active recent research in the area, the best-known upper and lower bounds on the deterministic query complexity (with query access to directions of edges) of finding a king in a tournament on $n$ vertices are from over 20 years ago, and the bounds do not match: the best-known lower bound is $\Omega(n^{4/3})$ and the best-known upper bound is $O(n^{3/2})$ [Shen, Sheng, Wu, SICOMP'03]. Our contribution is to show \emph{tight} bounds (up to logarithmic factors) of $\widetilde{\Theta}(n)$ and $\widetilde{\Theta}(\sqrt{n})$ in the \emph{randomized} and \emph{quantum} query models, respectively. We also study the randomized and quantum query complexities of finding a maximum out-degree vertex in a tournament.
\end{abstract}

\section{Introduction}
A tournament is a complete directed graph. Many important properties of tournaments were studied by Landau~\cite{landau1953dominance} in the context of modelling dominance relations among a flock of chickens. Relevant to our paper is the notion of a \emph{king} in a tournament. This notion was defined by Maurer~\cite{maurer1980king}, also with the goal of identifying a reasonable measure of dominance to identify a `most dominant' chicken in a flock. Soon after Maurer's article, Reid~\cite{Reid82} showed existence of tournaments in which all vertices are kings. Tournaments also arise naturally in social choice theory where directions of edges depict preferences. A large amount of work has been devoted to defining a notion of a `winner' in a tournament, and determining the complexity of finding such winners. For instance, Dey~\cite{Dey17} studied the complexity of certain tournament solutions with motivations from social choice theory. The monograph by Moon~\cite{moon2015topics} sparked a line of research on tournaments and their structural properties.

A natural computational model to study the complexity of computing specific properties of a tournament, or more generally, a graph, is that of \emph{query complexity}. In this setting an algorithm may query presence/directions of edges in an unknown input graph. The goal is to minimize the number of such queries made in the worst case. There is a rich literature on query complexity of graph problems, starting over 50 years ago~\cite{rosenberg1973time, rivest1976recognizing, yao1987lower, hajnal1991omega, chakrabarti2001improved, durr2006quantum, childs2012quantum, Dey17}. The famous Aanderaa-Karp-Rosenberg conjecture~\cite{rosenberg1973time} or evasiveness conjecture posits that the query complexity of any non-trivial monotone graph property on $n$-vertex graphs has maximal deterministic query complexity, i.e., $\binom{n}{2}$. While the deterministic and randomized variants of this conjecture are wide open, the quantum version was recently resolved in the positive~\cite{ABKRT21} using Huang's breakthrough resolution of the sensitivity conjecture~\cite{huang2019induced}.

Our work deals with the query complexity of certain graph problems. In the next section, we describe the main graph problem of interest to us, and prior work on it.

\subsection{Related work}\label{sec: related}

It is well known that every tournament has at least one vertex $v$ such that every other vertex is reachable from $v$ via a path of length at most 2 (see Lemma~\ref{lemma: king in tournament} for a proof). Such a vertex $v$ is called a \emph{king} in the underlying tournament. Formally, one may define the following relation that captures this definition.
\begin{defi}[Kings in a tournament]\label{defi: king}
    Let $T$ be a complete directed graph on $n$ vertices. Identify the orientation of the tournament with a string $T$ in $\zone^{\binom{n}{2}}$, one variable $(\cbra{i, j}$ with $i \neq j \in [n])$ per edge (between vertex $i$ and vertex $j$) defining its direction. Define the relation $\king_n \subseteq \zone^{\binom{n}{2}} \times [n]$ by
    \[
    (T, v) \in \king_n ~\textnormal{if}~\forall u \in [n], ~\textnormal{either}~ v \rightarrow u ~\textnormal{or}~\exists w \text{ such that } v \rightarrow w \rightarrow u.
    \]
    Here the directions of the edges $v \rightarrow u$ and $v \rightarrow w \rightarrow u$ are as in $T$.
\end{defi}
A natural question arises: what is the query complexity of finding a king in an $n$-vertex tournament? The study of this was initiated by Shen, Sheng and Wu~\cite{SSW03}. They showed an algorithm with query complexity $O(n^{3/2})$ and also showed a non-matching lower bound of $\Omega(n^{4/3})$. For the upper bound, they crucially used the fact that a king in an in-neighbourhood of an arbitrary subset of vertices is also a king in the original tournament (see Lemma~\ref{lem: maurer king in inneighbour}). An outline of their upper bound is as follows: first arbitrarily choose a sub-tournament of a fixed size and find the maximum-out-degree vertex in it by querying all edges in this sub-tournament. Remove this vertex along with its out-neighbours, and proceed iteratively. When the number of remaining vertices is small enough, find a king using brute force (query all the edges in the remaining sub-tournament). Simple manipulation of parameters gives an upper bound of $O(n^{3/2})$. For the lower bound they design an adversary who answers an algorithm's queries using a fixed strategy, and show that every algorithm must make $\Omega(n^{4/3})$ queries in the worst case.  Ajtai et al.~\cite{AFHN15} independently showed the same bounds, in a different context. Despite active recent research in the area (see the next paragraph), these bounds from over 20 years ago remain state-of-the-art. It can be shown that a vertex with maximum out-degree is a king (see Lemma~\ref{lemma: king in tournament} and its proof). However, finding a vertex with maximum out-degree is known to be hard: it has deterministic query complexity $\Omega(n^2)$~\cite{balasubramanian1997finding}.

Biswas et al.~\cite{BJRS22} recently showed that the adversary used by~\cite{AFHN15, SSW03} to show an $\Omega(n^{4/3})$ lower bound cannot be used to prove a stronger lower bound. They additionally showed a query complexity upper bound of $O(n^{4/3})$ on finding a vertex from which at least half of all vertices are reachable by paths of length at most 2. They also considered variants of kings, and the complexity of finding such vertices. In a more recent work, Lachish, Reidl and Trehan~\cite{LRT22} showed an $O(n^{4/3})$-query algorithm to find a vertex from which at least $(\frac12 + \frac{2}{17})$ of the vertices are reachable by paths of length at most 2.

\subsection{Our contributions}\label{sec: contribs}
While the question of pinning down the deterministic query complexity of finding a king has been open and unimproved since the work of Shen, Sheng and Wu~\cite{SSW03}, the corresponding question in the randomized and quantum query models does not seem to have been studied in the literature. Our contribution is to give tight bounds in these models. We refer the reader to Section~\ref{sec: prelims} for a formal description of these models. Let $\sR(\cdot)$ and $\sQ(\cdot)$ denote bounded-error randomized and quantum query complexity, respectively. Our main theorems are as follows.
\begin{theorem}
    For all positive integers $n$,
    \begin{align*}
        \sR(\king_n) &= O(n \log \log n), \qquad \sR(\king_n) = \Omega(n),\\
        \sQ(\king_n) &= O(\sqrt{n}~\polylog(n)), \qquad \sQ(\king_n) = \Omega(\sqrt{n}).
    \end{align*}
\end{theorem}
We mentioned earlier that a vertex of maximum out-degree in a tournament is a king, and finding a vertex of maximum out-degree is known to have deterministic query complexity $\Omega(n^2)$. We show that even the randomized query complexity is $\Omega(n^2)$, and we also show bounds for the quantum query complexity of this task. Define the relation $\MOD_n \subseteq \zone^{\binom{n}{2}} \times [n]$ to consist of the elements $(T, v)$ where $v$ is a maximum out-degree vertex in the $n$-vertex tournament described by $T$.
\begin{theorem}
    For all positive integers $n$,
    \begin{align*}
        \sR(\MOD_n) & = \Theta(n^2),\qquad
        \sQ(\MOD_n) = O(n^{3/2}), \qquad \sQ(\MOD_n) = \Omega(n).
    \end{align*}
\end{theorem}
We suspect that $\sQ(\MOD_n) = \Theta(n^{3/2})$, but we leave open the problem of closing the gap between the upper and lower bounds in the quantum setting.
\subsubsection{Sketch of randomized upper bound for finding a king}
As mentioned in Section~\ref{sec: related}, the upper bound of Shen, Sheng and Wu crucially uses the fact that a king in the in-neighbourhood of an arbitrary vertex is also a king in the original tournament (Lemma~\ref{lem: maurer king in inneighbour}). A simple counting argument shows that a \emph{uniformly random} vertex in an $n$-vertex tournament has out-degree $\Omega(n)$ with high probability. This suggests a natural randomized iterative algorithm: in each step sample a few vertices and query all edges incident on them, until a vertex with large out-degree in the current sub-tournament is seen. We then remove this vertex along with all its out-neighbours from the tournament, and iterate. Since a random vertex has out-degree that is linear in the number of vertices with high probability, this process results in a small sub-tournament (with at most $\sqrt{n}$ vertices) after $O(\log n)$ iterations. At this point we can afford to query the entire remaining sub-tournament to find a king in it, and it can be shown by applying Lemma~\ref{lem: maurer king in inneighbour} iteratively that this vertex is also a king the original tournament.

\subsubsection{Sketch of quantum upper bound for finding a king}
Our quantum algorithm follows the same structure as our randomized one, but we run into some issues during a naive simulation. The following are the issues, along with how we handle them:
\begin{itemize}
    \item When trying to sample a vertex with high out-degree, we cannot afford to query all edges incident on a vertex to compute its out-degree since our algorithm needs to have query complexity essentially $O(\sqrt{n})$. To circumvent the need of querying all edges incident on a vertex to compute its in-degree, we use the subroutine of \emph{approximate counting}~\cite{brassard2002quantum} that returns an approximation of the in-degree but offers a quadratic speedup. It may seem like one could use a classical algorithm for approximate counting here, but such a classical sampling-based algorithm would require $\widetilde{\Omega}(n)$ queries if the number of in-neighbours is small, say $\polylog(n)$ (see the fourth bullet as to why such a case may arise).
    \item A second issue that arises is when we need to sample a vertex from the current sub-tournament. It is no longer clear how to do this in the quantum setting since we do not explicitly know the vertices remaining. However, we keep track of the set of vertices $W$ whose out-neighbours have effectively been removed in previous iterations. The number of iterations of the algorithm, and hence the size of $W$, is bounded by $O(\log n)$. We then use key properties of Grover's search algorithm: first set up a uniform superposition over all vertices. Next we `mark' the vertices in the current sub-tournament (call such vertices `good') using $O(\log n)$ queries: for a given vertex we only need to check if it is an out-neighbour of any of the vertices in $W$. Making these queries in superposition allows us to mark the good vertices in $O(\log n)$ queries. We then apply the Grover iterate a suitable number (at most $O(\sqrt{n})$) of times. At this point we make a measurement in the computational basis: by the correctness of Grover's algorithm, the probability of seeing a good vertex (i.e., a vertex in the current sub-tournament) is large.  The structure of Grover's algorithm implies that conditioned on not seeing a bad vertex, each good vertex is seen with equal probability. This effectively simulates sampling a uniformly random vertex from the current sub-tournament.
    \item Having sampled a vertex $v$ from the in-neigbourhood of $W$, the randomized algorithm next computes the out-degree of $v$ in the in-neighbourhood of $W$. We cannot afford to do this exactly since we do not explicitly know the in-neighbourhood of $W$, and moreover it may be very large. We are able to get around this using similar ideas to that in the previous bullet.
    \item Finally, it is no longer clear how to do the final brute-force step in the last remaining sub-tournament since we do not explicitly know the remaining $\sqrt{n}$ vertices. To handle this, we first modify the randomized algorithm so as to only have $O(\polylog(n))$ vertices remaining in this `brute-force' step, while still having only $O(\log n)$ iterations overall. Thus, the query complexity so far is still $\widetilde{O}(\sqrt{n})$. We can then use Grover's search repeatedly (or an improvement thereof, Theorem~\ref{theo: Grover search find k elements}) to find all the remaining vertices with high probability in $\widetilde{O}(\sqrt{n})$ queries. At this point we can find a king in the remaining sub-tournament using $O(\polylog(n))$ queries. By the same argument as in the randomized case, this vertex is also a king in the original tournament.
\end{itemize}

\subsubsection{Sketch of lower bounds for finding a king}
To show our lower bounds, we restrict our attention to a special class of tournaments, described below. Fix an arbitrary tournament $T$ on $n$ vertices where vertex $n$ is a source. This immediately implies that vertex $n$ is the unique king. For each $i \in [n-1]$, define the tournament $T_i$ to be $T$ with edges incident on vertex $i$ flipped so as to make vertex $i$ the source. Note that these sets of edges are disjoint for every $i \neq j$. If we assign one variable to each such set and promise that at most one of them has value 1 (i.e., has edges in the opposite direction from those in $T$), an algorithm that finds a king in these tournaments (which are unique by construction) also solves the Search problem on $n-1$ input variables. Our lower bounds of $\Omega(n)$ and $\Omega(\sqrt{n})$ on the randomized and quantum query complexities, respectively, then immediately follow from corresponding well-known lower bounds on the complexity of the Search problem.

\subsubsection{Sketch of bounds for finding a maximum out-degree vertex}
In the randomized setting, we use Yao's minimax principle (Lemma~\ref{lem: yao}). By this principle, it suffices to exhibit a hard distribution on input tournaments such that any deterministic algorithm with small query complexity must make large error when inputs are drawn from this distribution. We now describe the distribution: fix an $n$-vertex regular  tournament, say $T$ with $n$ odd and each vertex having out-degree exactly $(n-1)/2$ (such a tournament is easy to construct iteratively, for example) and flip a uniformly random edge of $T$. This causes a unique vertex of the new tournament to be a maximum out-degree vertex. Intuitively, finding this vertex is as hard as finding the edge that has been flipped, and it is well known that searching for a marked element among $k$ elements has randomized query complexity $\Omega(k)$. We formalize this argument in Theorem~\ref{thm: mod}. Our quantum lower bound uses similar ideas, and involves a reduction from the Search problem on an $\binom{n}{2}$-bit string, which has quantum query complexity $\Omega(n)$~\cite{boyer1998tight}. For the quantum upper bound, we use a maximum finding routine over the degrees of the vertices. Each degree can be computed using $n-1$ queries, and the maximum can be found in $O(\sqrt{n})$ queries~\cite{durr1996quantum}, giving us an $O(n^{3/2})$ upper bound.

\section{Preliminaries}\label{sec: prelims}
All logarithms in this paper are base 2. We use the notation $\polylog(n)$ to denote a quantity that is $\log^{c}n$ for a constant $c > 0$ (independent of $n$). For a positive integer $n$, we use the notation $[n]$ to denote the set $\cbra{1, 2, \dots, n}$. For an event $X$, let $\II[X]$ denote the indicator of $X$, i.e., $\II[X] = 1$ if $X$ occurs, and $\II[X] = 0$ if $X$ does not occur.

\subsection{Tournaments}
A tournament $T$ on a vertex set $V$ is a complete graph such that each edge is directed. Throughout this paper, unless mentioned otherwise, we consider tournaments $T$ on $n$ vertices and denote the vertex set by $V = [n]$. Such  a tournament has $\binom{n}{2}$ directed edges. We identify an $n$-vertex tournament with a binary string in $\zone^{\binom{n}{2}}$: an element of $[n]$ corresponds to the label of a vertex, and there is one variable ($\cbra{i, j}$ with $i \neq j \in [n]$) per edge (between vertex $i$ and vertex $j$) that defines its direction. For a tournament $T$ and vertex $v \in V$, let $N^-(v)$ denote the set of in-neighbours of $v$,
i.e., $N^-(v) = \cbra{u \in [n] \setminus \cbra{v} | u \rightarrow v \textnormal{ is an edge in }T}$ and let $N^+(v)$ denote the set of out-neighbours of $v$ (i.e., $\cbra{u \in [n] \setminus \cbra{v} | v \rightarrow u \textnormal{ is an edge}})$. Also, let $d^+(v) = |N^+(v)|$ and $d^-(v) = |N^-(v)|$ denote the out-degree and in-degree of $v$, respectively.
Since $T$ is a tournament, $d^+(v) + d^-(v) = (n-1)$ for all $v \in V$. For $S \subseteq V$, let $T[S]$ be the tournament induced on the vertices in $S$.
For a subset $W \subseteq V$, define $W^- = \{v \in V \mid v \rightarrow w \text{ is an edge for all } w \in W\}$. If $W = \emptyset$ then define $W^- = V$.
A vertex $v \in V$ is a \textit{king} if every vertex in $V \setminus\{v\}$ is reachable from $v$ by a path of length at most $2$. This is formally captured in Definition~\ref{defi: king} and repeated below for convenience.
Define the relation $\king_n \subseteq \zone^{\binom{n}{2}} \times [n]$ by $(G, v) \in \king_n ~\textnormal{if}~\forall u \in [n] \setminus \cbra{v}, ~\textnormal{either}~ v \rightarrow u ~\textnormal{or}~\exists w: v \rightarrow w \rightarrow u$.
Here the directions of the edges $v \rightarrow u$ and $v \rightarrow w \rightarrow u$ are as in the tournament $G$.
A well-known fact about tournaments is that every tournament has a king. We give a proof for completeness.

\begin{lemma}[Folklore]\label{lemma: king in tournament}
    Let $T \in \zone^{\binom{n}{2}}$ be a tournament. Then there exists a vertex $v \in [n]$ such that $(T, v) \in \king_n$.
\end{lemma}
\begin{proof}
    Consider a vertex $v$ of maximum out-degree. We show that such a vertex is a king. Consider the partition of $V$ into three disjoint sets: $\{v\}$, $N^+(v)$ and $N^-(v)$. Clearly, every vertex in $N^+(v)$ is at a distance at $1$ from $v$. Towards a contradiction, assume that there is a vertex $w$ in $N^-(v)$ such that there is no path of length $2$ of the form $v \rightarrow u \rightarrow w$, for some $u \in N^+(v)$. Thus every vertex in $N^+(v)$ is an out-neighbour of $w$. Since $v$ is also an out-neighbour of $w$, the out-degree of $w$ is greater than that of $v$, which is a contradiction.
\end{proof}

The above lemmas shows that any vertex with maximum out-degree in a tournament is a king in that tournament. However, as discussed in Section~\ref{sec: related}, finding a vertex of maximum out-degree is known to be hard.
We need the following result due to~\cite{maurer1980king}.

\begin{lemma}[\cite{maurer1980king}]
\label{lem: maurer king in inneighbour}
Let $T \in \zone^{\binom{n}{2}}$ be a tournament and $v \in [n]$. If a vertex $u$ in $N^{-}(v)$ is a king in $T[N^{-}(v)]$, then $u$ is a king in $T$.
\end{lemma}

The proof of the above lemma is easy: If $u$ is a king of the tournament $T[N^-(v)]$, then every vertex in $N^-(v)$ is at a distance at most $2$ from $u$. Also, since $u$ is an in-neighbour of $v$, every vertex in $N^+(v)$ is at a distance $2$ from $u$.
We also need the following lemma from~\cite{maurer1980king}.

\begin{lemma}[\cite{maurer1980king}]
\label{lemma: degree sum in tournament}
    Let $T \in \zone^{\binom{n}{2}}$ be a tournament.
    $\sum_{i=1}^n d^+(i) = \sum_{i=1}^n d^-(i) = \binom{n}{2}$.
\end{lemma}

We also need the following observation on the structure of a tournament (see e.g.,~\cite{balasubramanian1997finding}).
\begin{lemma}
    \label{lemma: vertices with outdegree k}
    Let $T \in \zone^{\binom{n}{2}}$ be a tournament and $k \geq 0$. Then, the number of vertices $v$ such that $d^{+}(v) \leq k$ is at most $2k+1$. 
\end{lemma}

\subsection{Query complexity}\label{subsec: query comp}
A deterministic decision tree $T$ on $m$ variables is a binary tree where the internal nodes are labeled by variables and leaves are labeled with elements of a set $\cR$. Each internal node has a left child, corresponding to an edge labeled 0, and a right child corresponding to an edge labeled $1$. On an input $x \in \zone^m$, $T$'s computation traverses a path from root to leaf as follows. At an internal node, the variable associated with that node is \emph{queried}: if the value obtained is $0$, the computation moves to the left child, otherwise it moves to the right child. The output of $T$ on input $x$, denoted by $T(x)$, is the label of leaf node reached.
We say that a decision tree $T$ computes the relation $f \subseteq \zone^{m} \times \cR$ if $(x, T(x)) \in \cR$ for all $x \in \zone^m$.
The deterministic query complexity of $f$, is 
$\sD(f) := \min_{T : T \text{ computes } f} \textnormal{depth}(T)$.
A randomized decision tree $\mathcal{A}$ is a distribution $\mathcal{D}_{\mathcal{A}}$ over deterministic decision trees. On input $x \in \zone^m$, the computation of $\mathcal{A}$ proceeds by first sampling a deterministic decision tree $T$ according to $\mathcal{D}_{\mathcal{A}}$, and outputting the label of the leaf reached by $T$ on $x$. We say $\mathcal{A}$ computes $f$ with bounded error if for every input $x$, $\Pr[(x, \cA(x)) \in \cR] \geq 2/3$. The randomized query complexity of $f \subseteq \zone^m \times \cR$ is defined as follows.
$\cR(f) = \min_{\substack{\mathcal{A} \textrm{~computing~} f \\ \textrm{~with error~} \le 1/3}} ~\max_{T: \mathcal{D}_{\mathcal{A}}(T) > 0}\textnormal{depth}(T)$.

\subsection{Preliminaries for quantum query complexity}
\label{section: Preliminaries for quantum query complexity}
We refer the reader to~\cite{NC16, lecturenotes} for basics of quantum computing. A quantum query algorithm $\mathcal{A}$ computing a relation $f \subseteq \zone^m \times \cR$ begins in an input-independent initial state $\ket{\psi_0}$, applies a sequence of unitaries $U_0, O_x, U_1, O_x, \cdots, U_T$, and performs a measurement. Here, the unitaries $U_0, U_1, \dots, U_T$ are independent of the input. The unitary operation $O_x$ represents the `query' operation, and maps $\ket{i}\ket{b}$ to $\ket{i}\ket{b \oplus x_i}$ for all $i \in [m]$ and $\ket{0}$ to $\ket{0}$.
We say that $\mathcal{A}$ is a bounded-error algorithm computing $f$ if for all $x \in \zone^m$, the probability of outputting $b \in \cR$ such that $(x, b) \in f$ is at least $2/3$. The bounded-error quantum query complexity of $f$, denoted by $\sQ(f)$, is the least number of queries required for a quantum query algorithm to compute $f$ with error at most $1/3$.

We also need some basic notions from Grover's search algorithm~\cite{Grover96}, a fundamental quantum algorithm, referring the reader to~\cite[Chapter 7]{lecturenotes} for more details. In the search problem, a quantum algorithm is given quantum query access to a string $x \in \zone^n$. It is convenient to work with the `phase-query' unitary $O_{x, \pm}$ which satisfies $O_{x, \pm} \ket{i} = (-1)^{x_i} \ket{i}$. 
The goal is to find an $i \in [n]$ such that $x_i = 1$ with probability at least $2/3$ if such an $i$ exists, otherwise return that there is no such element. An $i$ which satisfies $x_i = 1$ is also called a \textit{marked element} and thus the goal is to find a marked element with high probability, if such an element exists.

Let $t := |\{i \in [n] : x_i = 1\}|$. Grover's algorithm starts with the uniform superposition $\ket{U} = \frac{1}{\sqrt{n}} \sum_{i = 1}^{n} \ket{i}$,
and proceeds by applying Grover's iterate (which is an application of $O_{x, \pm}$ followed by
a reflection about $\ket{U}$) several times. After $k$ applications of Grover's iterate the resulting state is
\begin{align}
    \sin((2k+1)\theta) \sum_{i : x_i = 1} \ket{i} + \cos((2k+1)\theta) \sum_{i : x_i = 0} \ket{i}, \label{eqn: state after Grover iterate}
\end{align}
where $\theta = \arcsin(\sqrt{t/n})$. It is known that Grover's algorithm finds a marked element in $x$ (if it exists) with $O(\sqrt{n})$ applications of the query oracle $O_{x,\pm}$, and probability at least $2/3$.
Standard error reduction yields the following theorem.

\begin{theorem}
\label{theo: Grover search}
Given query access to $x \in \zone^n$, there is a quantum algorithm that decides whether the Hamming weight of $x$ is $0$ or returns an $i \in [n]$ such that $x_i = 1$, with error at most $\delta$. The query complexity of this algorithm is $O(\sqrt{n \cdot \log(1/\delta)})$.
\end{theorem}

Grover's algorithm is known to be asymptotically optimal.
\begin{theorem}[{\cite{boyer1998tight}}]\label{theo: grover tight}
    A quantum algorithm that solves the Search problem with error $2/5$ on $n$-bit inputs must have query complexity $\Omega(\sqrt{n})$, even when the inputs are promised to have Hamming weight either 0 or 1.
\end{theorem}
The following theorem, due to D{\"u}rr and H{\o}yer~\cite{durr1996quantum}, is a generalization of Grover's search algorithm, to find the maximum number in an input list.
\begin{theorem}[{\cite{durr1996quantum}}]\label{theo: quantum maximum finding}
    Let $T$ be an unsorted table of $n$ items. There exists a quantum query algorithm of cost $O(\sqrt{n})$ that has query access to $T$ and returns the maximum element of $T$ with probability at least $2/3$.
\end{theorem}
We require the following theorem, essentially due to Boyer et al.~\cite{boyer1998tight}.\footnote{Their bound is for bounded-error algorithms and does not have polylogarithmic factors in the query complexity. Standard error reduction gives us Theorem~\ref{theo: Grover search find k elements}.}
\begin{theorem}[\cite{boyer1998tight}]
\label{theo: Grover search find k elements}
Given query access to $x \in \zone^n$ with $|x| \geq k$, there is a quantum algorithm that outputs, with query complexity $O(\sqrt{(n/k)} \log(1/\delta))$ and error probability at most $\delta$, an index $i \in [n]$ with $x_i = 1$.
\end{theorem}
We obtain the following immediate corollary by repeating the algorithm in Theorem~\ref{theo: Grover search find k elements} $k$ times and updating the `marked' elements after each application.
\begin{corollary}\label{cor: multiple grover}
Given an input parameter $k$ and query access to $x \in \zone^n$, there is a quantum algorithm that does the following with query complexity $O(\sqrt{nk} \log\log(n))$ and error probability at most $1/\polylog(n)$:
\begin{itemize}
    \item If $|x| \geq k$, it returns $k$ distinct indices $i_1, \dots, i_k \in [n]$ such that $x_{i_j} = 1$ for $j \in [k]$.
    \item If $|x| < k$, it outputs all indices $i$ with $x_i = 1$, along with the information that $|x| < k$.
\end{itemize}
\end{corollary}

Our quantum algorithm also uses quantum approximate counting as a sub-routine. Here, an algorithm is given query access to a string $x \in \zone^n$. The indices $i \in [n]$ such that $x_i = 1$ are again called `marked'. For an input parameter $\epsilon$ the goal of the algorithm is to output a multiplicative $(1 \pm \epsilon)$-approximation of the number of marked indices of $x$. An optimal quantum algorithm for approximate counting was first given by Brassard et al.~\cite{brassard2002quantum}. We use a version due to Aaronson and Rall~\cite{AR21}.

\begin{theorem}[\cite{AR21}]
\label{theo: quantum approx counting Aaronson Rall}
    There exists a quantum algorithm that, given $\epsilon > 0$ and query access to a string $x \in \zone^n$, outputs an estimate $\tilde{K}$ of $K = |\{i : x_i =1 \}|$ such that $K(1-\epsilon) \leq \Tilde{K} \leq K(1+\epsilon)$ with probability at least $(1-\delta)$. The query complexity of this algorithm is $O(\sqrt{n/K} \cdot 1/\epsilon \cdot \log(1/\delta))$.
\end{theorem}

\section{Randomized algorithm}
Throughout this section and the next, unless mentioned otherwise, a tournament $T$ is assumed to be in $\zone^{\binom{n}{2}}$, and its vertex set is denoted by $V = [n]$. Query algorithms are assumed to have classical/quantum query access to the edge directions of $T$, that is, the individual bits of the corresponding $\binom{n}{2}$-bit string.

In this section we give a randomized algorithm for finding a king in a tournament $T \in \zone^{\binom{n}{2}}$ with query complexity $O(n \log \log n)$ and success probability at least $2/3$. First, we make the following simple observation, which shows that a randomly chosen vertex from $V = [n]$ has a large number of out-neighbours with high probability.

\begin{lemma}[Out-degree of a random vertex is large]
\label{lemma: Out-degree of a random vertex}
    For all positive integers $n$, a tournament $T \in \zone^{\binom{n}{2}}$ and a vertex $v \in V$ chosen uniformly at random, $d^{+}(v) \geq \lfloor(n-1)/5\rfloor$ with probability at least $3/5$.
\end{lemma}
\begin{proof}
From Lemma~\ref{lemma: vertices with outdegree k}, $| \{ v \in V \mid d^+(v) < \lfloor (n-1)/5 \rfloor\}| \leq 2((n-1)/5 - 1) +1 = (2n-7)/5 < 2n/5$.
Thus, the fraction of vertices with out-degree at least $\lfloor (n-1)/5 \rfloor$ is at least $3/5$.  
\end{proof}

Lemma~\ref{lemma: Out-degree of a random vertex} suggests a natural randomized query algorithm, given in Algorithm~\ref{algo: randomized query}. We show in Theorem~\ref{theo: Randomized query upper-bound} that the algorithm makes $O(n \log \log n)$ queries to $T$ in the worst case, and returns a king with probability at least $2/3$.

\begin{algorithm}[h]
\begin{algorithmic}[1]

\State \textbf{Input:} Query access to edge directions of a tournament $T \in \zone^{\binom{n}{2}}$ where $V = [n]$.
\While{$|V| \geq \sqrt{n}$}\label{line: classical while}
    \State{$t \gets |V|$, $k \gets \lceil\log\log n\rceil$}\label{line: classical k logg t}
    \State{$v_1, \dots, v_k \gets $vertices chosen uniformly at random from $V$}
    \State{$w \gets \argmax_{u \in \cbra{v_1, \dots, v_k}}d^+(u)$}\Comment{querying all edges incident on $\cbra{v_1, \dots, v_k}$ in $T[V]$ and breaking ties arbitrarily}\label{line: classical query to look for highdeg}
    \If{$d^+(w) = t-1$}
        \State{Return $w$}\label{line: classical trivial terminate}
    \ElsIf{$d^+(w) < \lfloor(t-1)/5\rfloor$}\label{line: classical bad}
        \State{Return a random vertex $v \in V$}\label{line: classical return random}
    \Else\label{line: classical else}\Comment{$\lfloor(t-1)/5\rfloor \leq d^+(w) < t-1$ here}
        \State{$V \gets N^-(w)$}\Comment{This is the in-neighbourhood of $w$ in the set $V$, and not in the whole vertex set $[n]$.}
        \State \textbf{continue}\label{line: classical continue while loop}
    \EndIf
\EndWhile
\State{$w \gets $ a king in $T[V]$}\label{line: classical brute force}\Comment{query all edges in the sub-tournament $T[V]$}
\State{Output $w$}\label{line: classical end}
\caption{Randomized Query Algorithm}
\label{algo: randomized query}
\end{algorithmic}
\end{algorithm}

\begin{theorem}
\label{theo: Randomized query upper-bound}
    Let $n > 0$ be a positive integer. Then,
    $\sR(\king_n) = O(n \log\log n)$.
\end{theorem}
\begin{proof}
    Consider Algorithm~\ref{algo: randomized query}. 
    We first analyze the query cost of the algorithm. For the correctness, we define `bad events', argue correctness of the algorithm conditioned on no bad event occurring, and then upper bound the probability of a bad event happening.

\paragraph*{Query complexity}
In order to upper bound the query complexity, first note that each iteration of the \textbf{while} loop (Line~\ref{line: classical while}) uses $k \cdot |V| \leq |V| \log \log n$ queries in the worst case. Furthermore, the \textbf{while} loop goes into the next iteration (Line~\ref{line: classical continue while loop}) if and only if $|V| > \sqrt{n}$ (Line~\ref{line: classical while}) and a vertex $w$ of out-degree at least $\lfloor(t-1)/5\rfloor$ has been found in Line~\ref{line: classical query to look for highdeg} (see comment on Line~\ref{line: classical else}). This means that the size of the vertex set reduces by a factor of at least $4/5$ in the next iteration of the \textbf{while} loop. In particular, this means in the $i$'th iteration of the \textbf{while} loop, we have $|V|\leq (4/5)^i \cdot n$, and thus there are $O(\log n)$ iterations of the \textbf{while} loop in the worst case. Finally, Line~\ref{line: classical brute force} accounts for at most $O(n)$ queries since $|V| < \sqrt{n}$ here. The worst-case query complexity is thus upper bounded by
\begin{align*}
    n + \sum_{i = 0}^{O(\log{n})} \left(\frac{4}{5}\right)^i \cdot n \cdot O(\log \log n) = O(n \log\log n).
\end{align*}

\paragraph*{Bad event, and correctness assuming no bad event} The event of Line~\ref{line: classical return random} occurring during the run (i.e., Line~\ref{line: classical bad} being triggered in any iteration) is defined to be the bad event.
Conditioned on the bad event not occurring, the algorithm either terminates on Line~\ref{line: classical trivial terminate} or Line~\ref{line: classical end}. Clearly when the algorithm terminates on Line~\ref{line: classical trivial terminate} or Line~\ref{line: classical end}, the output vertex is a king in the sub-tournament being considered at the moment. If the \textbf{while} loop has not even completed once, the current sub-tournament is the same as the original tournament, and we are done. If the \textbf{while} loop has completed at least once,  the sub-tournament being considered at the moment is the sub-tournament of a tournament $T'$ (which itself may be a sub-tournament of $T$) induced by the in-neighbourhood of a specific vertex. Applying Lemma~\ref{lem: maurer king in inneighbour}, we conclude that the king in the current sub-tournament is also a king in $T'$, and also the whole tournament by applying Lemma~\ref{lem: maurer king in inneighbour} repeatedly now. Hence conditioned on the bad event not occurring, the algorithm indeed outputs a correct answer.

\paragraph*{Probability of bad event} From Lemma~\ref{lemma: Out-degree of a random vertex}, the probability that Line~\ref{line: classical bad} is run in an iteration is at most $(2/5)^k \leq 1/\log^{\log 2.5} |V| \leq 1/\log^{1.3}n$. By a union bound, the probability that Line~\ref{line: classical bad} gets executed in any of the $O(\log n)$ iterations is at most $O(\log n)/\log^{1.3}(n) = o(1)$.
\end{proof}

\section{Quantum algorithm}\label{sec: quantum}
For $W \subseteq [n]$ and $v \in V$, we can decide whether $v$ is an out-neighbour of any $w \in W$ by making $|W|$ queries, by checking $x_{wv}$ for all $w \in W$. Similarly, $|W|$ queries are sufficient to decide whether $v$ is an in-neighbour of some vertex $w \in W$.
This simple classical algorithm can easily be simulated in the quantum setting, which gives us the following observation.

\begin{observation}
\label{obs: mark out neighbour}
    For a tournament $T \in \zone^{\binom{n}{2}}$ and a known subset of the vertices $W \subseteq V$, there exists a unitary transformation that
    maps the basis state $\ket{v}$ to $(-1)^{\II[v \in W^-]}\ket{v}$ using $|W|$ queries to $T$. In other words, there is a unitary transformation that has query cost $|W|$ and `marks' vertices in $W^-$.
\end{observation}

Before proving the main theorem of this section, we give two lemmas. The algorithm in these lemmas will be used in the proof of the main theorem.

\begin{lemma}
\label{lemma: algo in sample}
    Let $T \in \zone^{\binom{n}{2}}$ be a tournament, $W \subseteq V$ and $t = \Theta(\log\log n)$ be an integer.
    There exists a quantum algorithm \textnormal{In-Sample}$(T, W, t)$, Algorithm~\ref{algo: in-smaple}, that with error probability at most $1/(\polylog(n))$, returns a set of uniformly distributed and independent samples from $W^-$ of size $t$. 
    The query complexity of this algorithm is $O(|W| \cdot \sqrt{n} \cdot \polyloglog(n))$.
\end{lemma}

\begin{proof}
    Consider Algorithm~\ref{algo: in-smaple}. 
    We first analyze the query cost of the algorithm. For the correctness, we define `bad events', argue correctness of the algorithm conditioned on no bad event occurring, and then upper bound the probability of a bad event happening.

\paragraph*{Query complexity}
We upper bound the worst-case query complexity of the algorithm.
Line~\ref{line: in-sample: estimate} of the algorithm case costs $O(|W| \cdot \sqrt{N} \cdot \polylog(N))$ from Theorem~\ref{theo: quantum approx counting Aaronson Rall}.
The \textbf{while} loop from Line~\ref{algo: in-sample while loop} runs for $O(t \cdot \polyloglog(N)) = O(\polyloglog(N))$ times, and each Grover's iterate in each of these iterations makes $O(|W| \cdot \sqrt{N})$ queries in Line~\ref{algo: in-sample use the estimate}. Also, Line~\ref{algo:in-smaple: check in W-} uses $|W|$ many queries.
Thus, the overall query cost of the algorithm is upper bounded by $O(|W| \cdot \sqrt{N} \cdot \polyloglog(N))$. Since $N = 10^4n$, we have an upper bound of $O(|W| \cdot \sqrt{n} \cdot \polyloglog(n))$.

\paragraph*{Bad event, and correctness assuming no bad event} 
If the estimate in Line~\ref{line: in-sample: estimate} is incorrect or if the algorithm has reached Line~\ref{line: in-sample algo: returns incorrect answer} is not in $W^-$ then we say that a bad event has occurred for Algorithm~\ref{algo: in-smaple}. We assume that these events have no happened.
Thus the estimate in Line~\ref{line: in-sample: estimate} is correct then $\Tilde{w}$ satisfies
    \begin{align*}
        |W^-| (1 - 1/100)\leq \Tilde{w}\leq |W^-| (1 + 1/100).
    \end{align*}
Define $w' = \floor{\Tilde{w}/2}$, thus $w'$ satisfies the following equations.
    \begin{align}
        |W^-|/4 & \leq w' \leq |W^-|, \nonumber\\
        1/2 \cdot \sqrt{|W^-|/N} & \leq \sqrt{w'/N} \leq \sqrt{|W^-|/N}\label{eqn: angles are close}.
    \end{align}
Let $x = |W^-|/N$. Since $|W^-| \geq 0$ and $|W^-| \leq n$, we have
\begin{align*}
    0 \leq x & \leq 1/10^4.
\end{align*}
Let $C = 1/10^4$.
For $x \in [0, \sqrt{C}]$ and $A \geq 1$ (whose value is to be fixed later), define
\begin{align*}
    g(x) &= A \arcsin{x/2} - \arcsin{x}.
\end{align*}
The derivative of $g$ is given by
\begin{align*}
    g'(x) 
    &= \frac{A/2}{\sqrt{1 - x^2/4}} - \frac{1}{\sqrt{1 - x^2}} \\
    &\geq \frac{A }{\sqrt{4 -  x^2}} - \frac{1}{\sqrt{1 - C}} \\
    &\geq A/2 - \frac{1}{\sqrt{1 - C}}.
\end{align*}
Thus for $A = 3\sqrt{1-C}$ the above derivative is positive for all $x \in [0,\sqrt{C}]$. Since $g(0) = 0$, we have, for $A \arcsin{x/2} \geq \arcsin{x}$.

From monotonicity of $\arcsin$ in $[0,1]$ and Equation~\eqref{eqn: angles are close} we have
\begin{align}
    \arcsin(1/2 \cdot \sqrt{|W^-|/N}) & \leq \arcsin(\sqrt{w'/N}) \leq \arcsin(\sqrt{|W^-|/N}) \nonumber \\
    1/A \cdot \arcsin(\sqrt{|W^-|/N}) & \leq \arcsin(\sqrt{w'/N}) \leq \arcsin(\sqrt{|W^-|/N}). \label{eq: bounding arcsins}
\end{align} 

In Line~\ref{line: in-sample: set number of iterations} we choose $\Tilde{k}$ to be $\floor{\left(\frac{\pi}{400\arcsin{\sqrt{w'/N}}} + \frac{1}{2}\right)}$. From Equation~\eqref{eq: bounding arcsins} we have
\begin{align}
    \left(\frac{\pi}{400\arcsin{\sqrt{|W^-|/N}}} + \frac{1}{2}\right) \leq \left(\frac{\pi}{400\arcsin{\sqrt{w'/N}}} + \frac{1}{2}\right) \leq (A+1) \cdot \left(\frac{\pi}{400\arcsin{\sqrt{|W^-|/N}}}  + \frac{1}{2} \right).
\end{align}
which implies
\begin{align}
    \left(\frac{\pi}{400\arcsin{\sqrt{|W^-|/N}}} - \frac{1}{2}\right) \leq \floor{\left(\frac{\pi}{400\arcsin{\sqrt{w'/N}}} + \frac{1}{2}\right)} \leq (A+1) \cdot \left(\frac{\pi}{400\arcsin{\sqrt{|W^-|/N}}}  + \frac{1}{2} \right). \label{eq: bound on the number of iterations}
\end{align}
From  Equation~\eqref{eqn: state after Grover iterate},
if we apply Grover's iterate $k$
times then the resulting state in Line~\ref{algo: in-sample use the estimate} is of the following form:
    \begin{align}
        \beta \sum_{v \in W^-} \ket{v} + \sqrt{(1 - \beta^2)} \sum_{v \in W^+} \ket{v}, \label{eq: probaility In-Sample}
    \end{align}
where $\beta = \sin((2k + 1) \cdot \arcsin{\sqrt{|W^-|/N}})$. From Equation~\eqref{eq: bound on the number of iterations} we have
\begin{align*}
    \frac{\pi}{200} \leq (2\tilde{k} + 1) \cdot \arcsin{\sqrt{|W^-|/N}} \leq (A+1) \frac{\pi}{200} + (A+2) \arcsin(\sqrt{|W^-|/N}) < \pi/2,
\end{align*}
where the last inequality follows due to the choice of $A$ ($A \leq 3$) and since $\sqrt{|W^-|/N} \leq 1/100$. Thus after $\Tilde{k}$ iterations, $\beta^2= \sin^2((2\tilde{k} + 1) \cdot \arcsin{\sqrt{|W^-|/N}})$ is a constant smaller than $\pi/2$.

Since we have assumed that the bad event in Line~\ref{line: in-sample algo: returns incorrect answer} has not occurred, this means that $t$ sample obtained is in $W^-$. From Equation~\eqref{eq: probaility In-Sample} each vertex in $W^-$ has an equal probability of being sampled.
Clearly, for different iterations of the \textbf{while} loop in Line~\ref{algo: in-sample while loop} the samples are independent. Also, in this case the algorithm returns in Line~\ref{line: in-sample: return R} after $t$ iterations and hence $\Omega(t)$ uniformly distributed and independent samples from $W^-$ are returned.

\paragraph*{Probability of bad event} 
The probability of the bad event happening in Line~\ref{line: in-sample: estimate} by Theorem~\ref{theo: quantum approx counting Aaronson Rall} is $O(1/\polylog(n))$. 
To upper bound the probability of the algorithm reaching Line~\ref{line: in-sample algo: returns incorrect answer}, observe that with probability $\beta^2 = \Omega(1)$ (see Equation~\eqref{eq: probaility In-Sample}) a vertex sampled in Line~\ref{algo: in-sample measure} is in the set $W^-$. 
Thus the probability that after $O(t~\polyloglog(n))$, less than $t$ vertices are seen in $W^-$ is upper bounded by $O(1/\polylog (n))$ by a Chernoff bound.
\end{proof}

\begin{algorithm}[H]
\begin{algorithmic}[1]
\State{\textbf{Input:} Query access to the adjacency matrix of a tournament $T \in \zone^{\binom{n}{2}}$ where $V = [n]$, $W \subseteq V$ such that $|W^-| \geq \log^{100}n$, and $t \in \mathbb{N}$ such that $t = \Theta(\log\log n)$.}

\State{$N \gets 10^4n$}
\State{$\ket{\phi} \gets \sum_{i = 1}^{N}\frac{1}{\sqrt{N}}\ket{i}$} 
\Comment{$\ket{\phi}$ is used as the starting state in Line~\ref{line: in-sample: estimate} and Line~\ref{algo: in-sample use the estimate} with vertices in $W^- \subseteq [n]$ marked (by first checking if $j \in [N]$ satisfies $j \leq n$, and marking such a $j$ using $|W|$ queries).}

\If{$W = \emptyset$}
    \State{$S \gets t$ samples from uniform superposition over $V$}
    \State{Return $S$}
\Else
    \State{$\Tilde{w} \gets$ estimate of $|W^-|$ from Theorem~\ref{theo: quantum approx counting Aaronson Rall} with $\epsilon = 1/100, \delta =1/\polylog(N) = 1/\polylog(n)$}. \label{line: in-sample: estimate} 

    \State{$w' \gets \floor{\Tilde{w}/2}$}
    
    \State{$\Tilde{k} \gets $ $\floor{\left(\frac{\pi}{400\arcsin{\sqrt{w'/N}}} + \frac{1}{2}\right)}$} \label{line: in-sample: set number of iterations}

    \State{$R \gets \emptyset$}
    \State{$\mathsf{count} \gets 0$}
    \While{$\mathsf{count} < O(t~\polyloglog(n))$} \label{algo: in-sample while loop}
        \State{$\ket{\psi_i} \gets $ state obtained by applying Grover's iterate $\tilde{k}$ times on $\ket{\phi}$,
        with vertices in ${W}^-$ being the marked elements} \label{algo: in-sample use the estimate} 
        
        \State{$v_i \gets$ measurement outcome of $\ket{\psi_i}$ in computational basis} \label{algo: in-sample measure}

        \If{$v_i \in W^{-}$}\Comment{query edges between $v_i$ and $W$} \label{algo:in-smaple: check in W-}
            \State{$R \gets R \cup \cbra{v_i}$}
        \EndIf
        \State{$\mathsf{count} \gets \mathsf{count} + 1$}
        \If{$|R| = t$}\Comment{If we have collected enough samples}
    \State{Return $R$} \label{line: in-sample: return R}\Comment{This is a set of uniformly distributed and independent samples from $W^-$ of size $t$ (See Lemma~\ref{lemma: algo in sample})}
    \EndIf
    \EndWhile
\EndIf

\State{Return $[t]$} \label{line: in-sample algo: returns incorrect answer} \Comment{The algo makes error in this case.}
   
\caption{The In-Sample$(T, W, t)$ algorithm for sampling many uniformly independent samples from a subset of vertices}
\label{algo: in-smaple}
\end{algorithmic}
\end{algorithm}

\begin{lemma}
\label{lemma: algo high in degree test}
    Let $T \in \zone^{\binom{n}{2}}$ be a tournament, $W$ be a subset of $V$ satisfying $|W^-| \geq \log^{100}n$ and $u$ be a vertex in $V$.
    There exists a quantum algorithm \textnormal{Decide-High-Out-Degree}$(T, W, u)$, Algorithm~\ref{algo: decide high out degree}, that returns with error probability at most $1/(\polylog(n))$, $\mathsf{True}$ if the out-degree of $u$ in $W^-$ is at least $|W^-|/5$ and $\mathsf{False}$ if the out-degree of $u$ in $W^-$ is at most $|W^-|/10$. The query complexity of this algorithm is $O(|W| \cdot \sqrt{n}~\polylog(n))$.
\end{lemma}

\begin{proof}
Consider Algorithm~\ref{algo: decide high out degree}. We first analyze the query cost of the algorithm. For the correctness, we define a `bad event', argue correctness of the algorithm conditioned on the bad event not occurring, and then upper bound the probability of the bad event happening.

\paragraph*{Query complexity}
The only queries used are in Line~\ref{algo: decide-high-out-deg estimate W-} and Line~\ref{algo: decide-high-out-deg estimate 2} of the algorithm. The query cost of these steps are upper bounded by $O(|W| \cdot \sqrt{n} \cdot \polylog(n))$ by Theorem~\ref{theo: quantum approx counting Aaronson Rall}.

\paragraph*{Bad event, and correctness assuming no bad event} The only bad event for Algorithm~\ref{algo: decide high out degree} are that either the estimates Line~\ref{algo: decide-high-out-deg estimate W-} or Line~\ref{algo: decide-high-out-deg estimate 2} is incorrect. Let us assume that the bad event has not happened. Then
\begin{align*}
        (1-1/100)|W^-| \leq \Tilde{w}_1 \leq (1+1/100)|W^-|,
\end{align*}
and
\begin{align*}
        (1-1/100)|N^+(u) \cap W^-| \leq \Tilde{w}_2 \leq (1+1/100)|N^+(u) \cap W^-|.
\end{align*}
We have
\begin{align*}
    \frac{99}{101} \cdot \frac{|N^+(u) \cap W^-|}{|W^-|} \leq \frac{\Tilde{w}_2}{\Tilde{w}_1} \leq \frac{101}{99} \cdot \frac{|N^+(u) \cap W^-|}{|W^-|}.
\end{align*}
Thus if $|N^+(u) \cap W^-|/|W^-| \geq 1/5$ then $\Tilde{w}_2/\Tilde{w}_1 \geq 99/505$ and if $|N^+(u) \cap W^-|/|W^-| \leq 1/10$ then $\Tilde{w}_2/\Tilde{w}_1 \leq 101/990$.

\paragraph*{Probability of bad event} 
By Theorem~\ref{theo: quantum approx counting Aaronson Rall} and a union bound, the probability of the bad event is upper bounded by $O(1/\polylog(n))$.
\end{proof}

\begin{algorithm}[h]
\begin{algorithmic}[1]
\State{\textbf{Input:} Query access to the edge directions of a tournament $T \in \zone^{\binom{n}{2}}$ where $V = [n]$, $W \subseteq V$ such that $|W^-| \geq \log^{100}n$, and $u \in V$.}

    \State{$\Tilde{w}_1 \gets $ estimate of $|W^-|$ using Theorem~\ref{theo: quantum approx counting Aaronson Rall} with $\epsilon = 1/100, \delta = 1/\polylog(n)$.} \label{algo: decide-high-out-deg estimate W-}
    \Statex{}\Comment{Since the algorithm is given $W$ is input, it can decide whether $v \in W^-$ by making $|W|$ queries.}
    
    \State{$\Tilde{w}_2 \gets $ estimate of $|N^+(u) \cap W^-|$ using Theorem~\ref{theo: quantum approx counting Aaronson Rall} with $\epsilon = 1/100, \delta = 1/\polylog(n)$.}
    \label{algo: decide-high-out-deg estimate 2}

    \Statex{}\Comment{Note that we do not have query access to the presence/absence of a vertex $v$ in $N^+(u) \cap W^-$. However such a query can be implemented with $1 + |W|$ queries: check if $v \rightarrow u$ is an edge, and check if $v \rightarrow w$ is an edge for any $w \in W$.}
\If{$\Tilde{w}_2/\tilde{w}_1 \geq 99/505$}
    \State{return \textsf{True}}
\Else
    \State{return \textsf{False}}
\EndIf
\caption{The Decide-High-Out-Degree$(T, W, u)$ subroutine}
\label{algo: decide high out degree}
\end{algorithmic}
\end{algorithm}

We now show our main result of this section.
\begin{theorem}
\label{theo: quantum algo main thm}
    Let $n > 0$ be a positive integer. Then $\sQ(\king_n) = O(\sqrt{n}~\polylog(n))$.
\end{theorem}

\begin{proof}
Consider Algorithm~\ref{algo: main quantum algo}. We first analyze the query cost of the algorithm. For the correctness, we define `bad events', argue correctness of the algorithm conditioned on no bad event occurring, and then upper bound the probability of a bad event happening.

\paragraph*{Query complexity}
First we upper bound $|W|$ at the end of the run of the algorithm. The \textbf{while} loop in Line~\ref{line: main quantum algo: main while loop} runs for at most $O(\log n)$ iterations. The algorithm starts with $W$ initialized to $\emptyset$ and is updated only in Line~\ref{line: main quantum algo: update W} where one new element is added to $W$.
Thus we have $|W| = O(\log n)$. 

Consider Line~\ref{line: main quantum algo: get U}. Since $|W| = O(\log n)$ and $k = \log^{100} n$, by Corollary~\ref{cor: multiple grover} the number of queries in this step is upper bounded by $O(|W|\sqrt{n}~\polylog(n)) = O(\sqrt{n}~\polylog(n))$, and thus the overall cost of queries executed in this line over at most $O(\log n)$ iterations is also $O(\sqrt{n}~\polylog(n))$.

In Line~\ref{line: main quantum algo: invoke in-sample}, the In-Sample algorithm (Algorithm~\ref{algo: in-smaple}) is called at most $O(\log n)$ times with $t = \Theta(\log\log n)$ and $|W| = O(\log n)$. Thus by Lemma~\ref{lemma: algo in sample}, the cost of this step is upper bounded by $O(|W|\sqrt{n}~\polylog(n)) = O(\sqrt{n}~\polylog(n))$.

Now consider the \textbf{for} loop in Line~\ref{line: main quantum algo: for loop}.
This loop is executed at most $O(\log n)$ times and each iteration of this loop invokes the algorithm Decide-High-Out-Degree, with $|W| = O(\log n)$, at most $|S|$ many times. 
Since $|S| = O(\polylog(n))$ (see Lemma~\ref{lemma: algo in sample}) the query cost in this loop is upper bounded by $O(|W| \cdot |S| \cdot \sqrt{n}~\polylog(n)) = O(\sqrt{n}~\polylog(n))$ in the worst case. 

The only remaining step in Line~\ref{line: quantum return king U}. In this case, since $|U| \leq \log^{100} n$ throughout the algorithm, at most $O(\polylog(n))$ queries are made.

\paragraph*{Bad event, and correctness assuming no bad event}
If any of the following events happen, we say that a bad event has happened for Algorithm~\ref{algo: main quantum algo}:
\begin{enumerate}[(I)]
    \item\label{item: I} The algorithm in Corollary~\ref{cor: multiple grover} which is used in Line~\ref{line: main quantum algo: get U} gives an incorrect answer.
    \item\label{item: II} The algorithm In-Sample (Algorithm~\ref{algo: in-smaple}) in Line~\ref{line: main quantum algo: invoke in-sample} fails to return a set of $\Omega(t) = \Omega(\log\log n)$ uniformly distributed and independent samples from $W^-$. 

    \item\label{item: III} The set $S$ obtained from In-Sample in Line~\ref{line: main quantum algo: invoke in-sample} does not contain a vertex of out-degree at least $|W^-|/5$ in $W^-$.

    \item\label{item: IV} The algorithm Decide-High-Out-Degree (Algorithm~\ref{algo: decide high out degree}) in Line~\ref{line: main quantum algo: invoke Decide-High-Out-Degree} returns \textsf{False}.
\end{enumerate}
We prove the correctness of the algorithm assuming that these bad events do not happen.
Consider the $j$'th iteration of the \textbf{while} loop in Line~\ref{line: main quantum algo: main while loop}, for $j \geq 1$, and let $W^{(j)}$ denote the set $W$ in this iteration.
$W^{(j)}$ is updated only in Line~\ref{line: main quantum algo: update W} by a $v$ which satisfies $v \in (W^{(j)})^-$. 
This is because each vertex of the set $S$ belongs to $W^-$ (see Line~\ref{algo:in-smaple: check in W-} of Algorithm~\ref{algo: in-smaple}). 
In the next iteration of the \textbf{while} loop, $(W^{(j+1)})^-$ is defined as $(W^{(j)})^- \cap N^-(v)$.
Thus by applying Lemma~\ref{lem: maurer king in inneighbour} iteratively, $(W^{(j+1)})^-$ contains a king in the tournament $T[(W^{(j)})^-]$, and hence a king in $T$.

Assuming that the bad events do not happen, we now argue that in $O(\log n)$ iterations the size of $W^-$ becomes smaller than $\log^{100} n$. 
In this case $U = W^-$ because of the property of Corollary~\ref{cor: multiple grover} used in Line~\ref{line: main quantum algo: get U}, and the algorithm correctly returns the king in Line~\ref{line: quantum return king U} by a similar argument as in the previous paragraph by iteratively applying Lemma~\ref{lem: maurer king in inneighbour}. 
The analysis is similar to that of proof of Theorem~\ref{theo: Randomized query upper-bound}. 
Since Decide-High-Out-Degree (Algorithm~\ref{algo: decide high out degree}) in Line~\ref{line: main quantum algo: invoke Decide-High-Out-Degree} does not return \textsf{False}, the out-degree of $v$ in $(W^{(j)})^-$ must be at least $|(W^{(j)})^-|/10$.

Thus $|(W^{(j+1)})^-| \leq (9/10) \cdot |(W^{(j)})^-|$, and after $O(\log n)$ iterations the size of $W^-$ is smaller than $\log n < \log^{100}n$.

\paragraph*{Probability of bad event} 
The probability of events~\ref{item: I},~\ref{item: II},~\ref{item: IV} are each upper bounded by $O(1/\polylog(n))$ by Corollary~\ref{cor: multiple grover}, Lemma~\ref{lemma: algo in sample} and Lemma~\ref{lemma: algo high in degree test}, respectively. The probability of event~\ref{item: III} conditioned on \ref{item: II} not happening is upper bounded by $(2/5)^{\Theta(\log\log(n))} = O(1/\polylog(n))$, thus the probability of event~\ref{item: III} is upper bounded by $O(1/\polylog(n))$.
The number of times that the events~\ref{item: I},~\ref{item: II},~\ref{item: III} can happen is at most $O(\log n)$, and~\ref{item: IV} can happen is at most $O(\polylog(n))$, a union bound implies the probability of a bad event happening is upper bounded by $O(1/\polylog(n))$.
\end{proof}

\begin{algorithm}[h]
\begin{algorithmic}[1]
\State{\textbf{Input:} Query access to the edge directions of a tournament $T \in \zone^{\binom{n}{2}}$ with $V = [n]$}

\State{$W \gets \emptyset$, $t \gets \Theta(\log\log n)$, and $\textsf{COUNT} \gets O(\log n)$}
\Statex{}\Comment{Recall that $\emptyset^- := V$}
\While{$\textsf{COUNT} > 0$} \label{line: main quantum algo: main while loop}
    \State{$\textsf{COUNT} \gets \textsf{COUNT} - 1$}

    \State{$U \gets $ the output of the algorithm in Corollary~\ref{cor: multiple grover}} with the string in $\zone^{[n]}$ as input where indices corresponding to vertices in $W^-$ are equal to 1 (marked), and $k = \log^{100} n$ \label{line: main quantum algo: get U}
    \Statex{}\Comment{query access to this string can be done using $|W|$ edge queries to $T$}

    \If{$|U| < \log^{100} n$}
        \State{\textbf{break}}\Comment{Go to Line~\ref{line: quantum return king U}}\label{line: quantum break}
    \Else
    \State{$S \gets \textnormal{In-Sample}(T, W, t)$} \label{line: main quantum algo: invoke in-sample}
    \Comment{We reach here if $|W^-|$~$\geq\log^{100} n$ (Line~\ref{line: quantum break} gets executed otherwise)}
  
    \State{$S' \gets S$}
    \For{$v \in S$} \label{line: main quantum algo: for loop}
        \State{$S' \gets S' \setminus \cbra{v}$}
        \If{Decide-High-Out-Degree$(T,W,v) == \mathsf{True}$}\label{line: main quantum algo: invoke Decide-High-Out-Degree}
        \Statex{}\Comment{Decide-High-Out-Degree can be applied since $|W^-| \geq \log^{100}n$} 
            \State{$W \gets W \cup \cbra{v}$} \label{line: main quantum algo: update W}
            \State{\textbf{break}}\Comment{Go to Line~\ref{line: main quantum algo: main while loop}}
        \EndIf
        \If{$S' == \emptyset$}
            \State{Return a random vertex $v \in V$}
        \EndIf
    \EndFor
    \EndIf
\EndWhile
\State{Return a king in $U$}\label{line: quantum return king U}\Comment{query all edges in $T[U]$}
\caption{Quantum Algorithm}
\label{algo: main quantum algo}
\end{algorithmic}
\end{algorithm}

\section{Lower bounds}
We show our lower bounds in this section. We first show our lower bounds for the query complexity of finding a vertex of maximum out-degree, and then our lower bounds for finding a king in a tournament.

\subsection{Maximum out-degree}
We show in this subsection that the randomized query complexity of finding a vertex of maximum out-degree in an $n$-vertex tournament is $\Omega(n^2)$. This task is formally defined as the relation $\MOD_n \subseteq \zone^{\binom{n}{2}} \times [n]$: $(G, v) \in \MOD_n~\textnormal{if}~d^{+}(v) \geq d^{+}(w)~\forall w\neq v \in [n]$.
Here the out-degrees of $v, w$ are according to the tournament $G$.
\begin{theorem}\label{thm: mod}
    For sufficiently large positive integers $n$, $\sR(\MOD_n) \geq n^2/100$.
\end{theorem}

We use Yao's minimax principle~\cite{yao1977probabilistic}, stated below in a form convenient for us.
\begin{lemma}[Yao's minimax principle]\label{lem: yao}
    For a relation $f \subseteq \zone^m \times \cR$, we have $\sR(f) \geq k$ if and only if there exists a distribution $\mu : \zone^m \to [0,1]$ such that $\sD_\mu(f) \geq k$.
    Here, $\sD_\mu(f)$ is the minimum depth of a deterministic decision tree that computes $f$ to error at most $1/3$ when inputs are drawn from the distribution $\mu$.
\end{lemma}

\begin{proof}[Proof of Theorem~\ref{thm: mod}]
Assume without loss of generality that $n$ is odd. We construct a hard distribution $\mu$ on $n$-vertex tournaments. We show that any deterministic query algorithm of cost less than $n^2/100$ must make error at least $1/3$ on inputs drawn from $\mu$, and this would prove the theorem by Yao's principle (Lemma~\ref{lem: yao}). Let $G$ be a fixed $n$-vertex regular tournament where every vertex has out-degree exactly $(n-1)/2$ (such a tournament is easy to construct, by induction, for example). The distribution $\mu$ is defined by taking $G$ and flipping the direction of a uniformly random edge. Note that all resultant tournaments have a unique vertex with maximum out-degree.

Consider a deterministic query algorithm (decision tree) that queries less than $n^2/100$ edges. Consider the leaf $L$ of this tree for which answers of all queries on its path are consistent with directions of edges in $G$. Say the label of this leaf is vertex $i$. Consider the set $S$ of all unqueried edges on the path to $L$ that are not incident on vertex $i$. We have $|S| \geq \binom{n}{2} - \frac{n^2}{100} - (n-1)$.
For each $e \in S$, the graph $G_e$ defined by flipping the direction of $e$ in $G$ reaches the leaf $L$. Moreover, the unique maximum out-degree vertex of $G_e$ is not vertex $i$ since $e$ is not incident on $i$ by the definition of $S$. This implies that the tree outputs the wrong answer on $G_e$. By the definition of $\mu$, we have $\mu(G_e) = 1/\binom{n}{2}$ for all $e \in S$. Thus, the mass of inputs under $\mu$ on which the decision tree makes an error is at least
\[
\sum_{e \in S}\mu(G_e) \geq \frac{\binom{n}{2} - \frac{n^2}{100} - n + 1}{\binom{n}{2}} > \frac{97}{100} > \frac{1}{3},
\]
where the second-to-last inequality holds for sufficiently large $n$. Lemma~\ref{lem: yao} yields the theorem.
\end{proof}

We now give our quantum bounds for $\MOD_n$.

\begin{theorem}\label{theo: MOD quantum bounds}
    For all positive integers $n$, $\sQ(\MOD_n) = O(n^{3/2}), \sQ(\MOD_n) = \Omega(n)$.
\end{theorem}
\begin{proof}
    For the upper bound we apply the maximum finding subroutine in Theorem~\ref{theo: quantum maximum finding} to the degree sequence of the input tournament. Finding the degree of a vertex (and hence a query of the maximum-finding algorithm) can be done with $n-1$ edge queries. Thus, this algorithm has cost $O(\sqrt{n}\cdot(n - 1)) = O(n^{3/2})$.

    For the lower bound, we give a reduction from the Search problem on an $\binom{n}{2}$-bit string. As in the proof of the randomized lower bound, assume $n$ is odd and let $G$ be a fixed $n$-vertex regular tournament where every vertex has out-degree exactly $(n-1)/2$. Towards a contradiction, suppose we have an algorithm $\cA$ that finds a maximum out-degree vertex in an $n$-vertex graph with query complexity $o(n)$ and probability at least $2/3$. We use $\cA$ to solve the Search problem on $\binom{n}{2}$-bit strings with the promise that the input has Hamming weight at most 1. On input $x \in \zone^{\binom{n}{2}}$ with $|x| \leq 1$, do the following:
    \begin{enumerate}
        \item Run the algorithm $\cA$ on the tournament $G \oplus x$. Here $G \oplus x$ denotes the bitwise XOR of $G$ and $x$. Suppose the output is $v \in [n]$.
        \item Run a $(99/100)$-error Search algorithm with query cost $O(\sqrt{n})$ on the $n-1$ indices of $x$ that are indexed by pairs with one element as $v$ (that is, indexed by the edges adjacent to $v$ in the corresponding tournament).
        \item Output the index returned by the search algorithm.
    \end{enumerate}
    The cost of this algorithm is clearly $o(n) + O(\sqrt{n})$. For the correctness, first note that when $|x| = 1$ and $G$ is such that all out-degrees are equal, the tournament $G \oplus x$ has exactly one maximum out-degree vertex. Thus, by the correctness of $\cA$, it outputs this vertex with probability at least $2/3$. Observe that the edge flipped in $G \oplus x$ from $G$ is adjacent to this vertex. In the event that the first step outputs the correct vertex, the edge that has been flipped in $G \oplus x$ from $G$ (i.e., the index $\cbra{i,j}$ with $x_{\cbra{i,j}} = 1$) is caught in the second step with probability at least $99/100$. Thus, this gives an algorithm solving the Search problem on $\binom{n}{2}$-bit strings with the promise that the input has Hamming weight at most 1, with success probability at least $(2/3)\cdot(99/100) > 3/5$. The query cost of this algorithm is $o(n)$ from the first step, by our assumption, and $O(\sqrt{n})$ from the second step. Thus the total cost is $o(n)$, which is a contradiction in view of Theorem~\ref{theo: grover tight}.
\end{proof}
We leave open the question of closing the gap in Theorem~\ref{theo: MOD quantum bounds}.

\subsection{Finding a king}
We show an $\Omega(n)$ lower bound for the randomized query complexity of finding a king in a tournament, and an $\Omega(\sqrt{n})$ quantum query lower bound. To show these lower bounds, we restrict our attention on input tournaments of a particular structured form that have the property that there is only one king (which is a source in the tournament). We then show a lower bound on the randomized and quantum query complexities of finding a king in these promised inputs, by a reduction from the Search problem on $n-1$ variables with the promise that the input has Hamming weight either 0 or 1, for which we know an $\Omega(n)$ lower bound in the randomized setting and an $\Omega(\sqrt{n})$ lower bound in the quantum setting. Our reductions use a simple modification of block sensitivity.

We require the following relation.
\begin{defi}
Let $n$ be a positive integer. Define the relation $\usearch_n \subseteq \zone^n \times \cbra{\emptyset} \cup [n]$ as 
$(0^n, \emptyset) \in \usearch_n$ and $(x, i) \in \usearch_n~\textnormal{when}~x = e_i$.
\end{defi}

\begin{claim}\label{claim: or reduction}
Let $n$ be a positive integer. Then,
\begin{align*}
\sR(\king_n) & \geq \sR(\usearch_{n-1}), \qquad
\sQ(\king_n) \geq \sQ(\usearch_{n-1}).
\end{align*}
\end{claim}

\begin{proof}
Consider an arbitrary input $x \in \zone^{\binom{n}{2}}$ such that the vertex $n$ is the source. For each $j \in [n-1]$, let $V_j \subseteq \left[\binom{n}{2}\right]$ be the set of edges incident on vertex $j$ that need to be flipped in the input $x$ to make vertex $j$ the source. We first make the following two observations:
\begin{align}
V_j \cap V_k & = \emptyset \quad \forall j \neq k \in [n-1], \qquad \bigcup_{i = 1}^{n-1} V_j = \left[\binom{n}{2}\right].
\end{align}
The first observation follows by considering an edge from vertex $\ell$ to vertex $m$. This edge only appears in $V_m$. Clearly every edge belongs to exactly one $V_j$, proving the second observation.

Using these two observations, the input set $\zone^{\binom{n}{2}}$ can also be expressed as $\zone^{V_1} \times \zone^{V_2} \times \cdots \times \zone^{V_{n-1}}$. For the remaining part of this proof we treat inputs to be of the latter form. In fact, we only restrict our attention to the case where each coordinate in a `block' has the same value.

For a string $y \in \zone^{n - 1}$, define the tournament
$x_y = \bigotimes_{i = 1}^{n - 1}y^{V_{i}}_i$.
Thus we have the following tournaments when $|y| \leq 1$:
\[
x_{e_j} = \begin{cases}0^{V_1} \times \cdots \times 0^{V_{j-1}} \times 1^{V_{j}} \times 0^{V_{{j + 1}}} \times \cdots \times 0^{V_{n-1}} & y = e_{j - 1}\\
0^{V_1} \times \cdots \times 0^{V_{n-1}} & y = 0^{n - 1}.
\end{cases}
\]

In other words, $x_{e_j}$ equals the tournament $x$ with variables in $V_{j}$ flipped, and $x_{0^{n - 1}} = x$.

Note that vertex $j$ is the source (and thus the unique king) in the tournament $x_{e_j}$. Thus, finding a king in the set of tournaments $\cbra{x_{e_j} : j \in [n-1]}$ is the same as finding a source in these tournaments. 
Thus, a query algorithm finding a king in the restricted input set $x_y : |y| \leq 1$ yields a query algorithm for $\usearch_{n-1}$ on input $y$, which proves the claim.
\end{proof}

From the well-known lower bounds of $\sQ(\usearch_{n-1}) = \Omega(\sqrt{n})$~\cite{BBBV97} and $\sR(\usearch_{n-1}) = \Omega(n)$, we obtain our main theorem of this section.
\begin{theorem}
Let $n$ be a positive integer and $\king_n \subseteq \zone^{\binom{n}{2}} \times [n]$. Then,
\begin{align*}
\sR(\king_n) = \Omega(n), \qquad \sQ(\king_n) = \Omega(\sqrt{n}).
\end{align*}

\end{theorem}

\bibliography{reference}

\end{document}